\documentclass[conference]{IEEEtran}
\IEEEoverridecommandlockouts
\usepackage{cite}
\usepackage{amsmath,amssymb,amsfonts}
\usepackage{algorithmic}
\usepackage{graphicx}
\usepackage{textcomp}
\usepackage{xcolor}
\usepackage{hyperref}
\usepackage{amsthm}        

\newtheorem{proposition}{Proposition}[section]
\usepackage{graphicx}      
\usepackage{booktabs}      
\usepackage{multirow}      
\usepackage{enumitem}      
\usepackage{cleveref}
\usepackage{subcaption}      

\newcommand{\ie}{\emph{i.e., }}
\newcommand{\eg}{\emph{e.g., }}

\newcommand{\cf}{\emph{cf. }}

\def\BibTeX{{\rm B\kern-.05em{\sc i\kern-.025em b}\kern-.08em
    T\kern-.1667em\lower.7ex\hbox{E}\kern-.125emX}}

\begin{document}

\IEEEoverridecommandlockouts 


\title{SpecFormer: Mitigating Embedding and Attention Collapse via Spectral-Aware Transformer for Recommendation}

\author{\IEEEauthorblockN{1\textsuperscript{st} Yu Cui $^{\ast}$\thanks{$^{\ast}$Equal contribution. $^{\dagger}$Corresponding authors.}}
\IEEEauthorblockA{\textit{Zhejiang University}\\
Hangzhou, China \\
cuiyu23@zju.edu.cn\\
ORCID: 0009-0001-6203-3022}

\\
\IEEEauthorblockN{4\textsuperscript{rd} Hao Zhang}
\IEEEauthorblockA{\textit{Alibaba Group} \\
Beijing, China \\
zh138764@alibaba-inc.com\\
ORCID:0009-0002-3137-4980}

\\
\IEEEauthorblockN{7\textsuperscript{th} Can	Wang}
\IEEEauthorblockA{\textit{Zhejiang University}\\
Hangzhou, China \\
wcan@zju.edu.cn \\
ORCID: 0000-0002-5890-4307}

\and
\IEEEauthorblockN{2\textsuperscript{nd} Yi Xu$^{\ast}$}
\IEEEauthorblockA{\textit{Alibaba Group} \\
Beijing, China \\
xy397404@alibaba-inc.com\\
ORCID:0009-0007-3571-8791}

\\
\IEEEauthorblockN{5\textsuperscript{rd} Yu Zhang}
\IEEEauthorblockA{\textit{Alibaba Group} \\
Beijing, China \\
daoji@alibaba-inc.com \\
ORCID:0000-0002-6057-7886}

\\
\IEEEauthorblockN{8\textsuperscript{rd} Jinxin Hu$^{\dagger}$}
\IEEEauthorblockA{\textit{Alibaba Group} \\
Beijing, China  \\
jinxin.hjx@alibaba-inc.com\\
ORCID:0000-0002-7252-5207}

\and
\IEEEauthorblockN{3\textsuperscript{rd} Jiahao Wang}
\IEEEauthorblockA{\textit{Alibaba Group} \\
Beijing, China \\
wjh177423@alibaba-inc.com\\
ORCID:0009-0000-3046-2351}

\\
\IEEEauthorblockN{6\textsuperscript{rd} Xiaoyi Zeng}
\IEEEauthorblockA{\textit{Alibaba Group} \\
Beijing, China \\
yuanhan@taobao.com\\
ORCID:0000-0002-3742-4910}

\\
\IEEEauthorblockN{9\textsuperscript{th} Jiawei Chen$^{\dagger}$}
\IEEEauthorblockA{\textit{Zhejiang University}\\
Hangzhou, China \\
sleepyhunt@zju.edu.cn\\
ORCID: 0000-0002-4752-2629}
}

\maketitle

\begin{abstract}
Transformer architectures have achieved remarkable success across diverse domains; however, directly applying their standard self-attention mechanism to recommendation often yields suboptimal performance, sometimes even trailing behind well-designed simple recommendation models. In this paper, we reveal that this performance bottleneck stems from severe embedding and attention collapse unique to recommendation scenarios. The heterogeneity and long-tail nature of recommendation data lead to a severe spectral collapse dominated by a few principal singular values. We further theoretically demonstrate that this triggers a vicious cycle in recommendation model's forward and backward propagation, which accelerates embedding and attention collapse and limits the model's scaling capability with increased depth. 

To address these issues, we propose SpecFormer, a novel Spectral-Aware Transformer designed for mitigating embedding and attention collapse in recommendation. Specifically, SpecFormer introduces 1) a Learnable Spectral Softening module to dynamically smooth the singular values distribution of the input token embeddings; 2) a Spectrum-softened Attention mechanism to model feature interaction under a more uniform spectral distribution space; 3) a Spectral Residual Position Encoding via Taylor expansion of singular values, explicitly providing a spectral inductive bias for feature interactions. Extensive experiments on one industrial and two public datasets demonstrate that SpecFormer significantly outperforms state-of-the-art baselines. Notably, SpecFormer has been successfully deployed in a real-world commercial recommender system and exhibits exceptional scaling capabilities: stacking SpecFormer layers actively improves the attention effective rank and recommendation performance.
\end{abstract}


\section{Introduction}

\begin{figure}[t]
    \centering
    \includegraphics[width=1\linewidth]{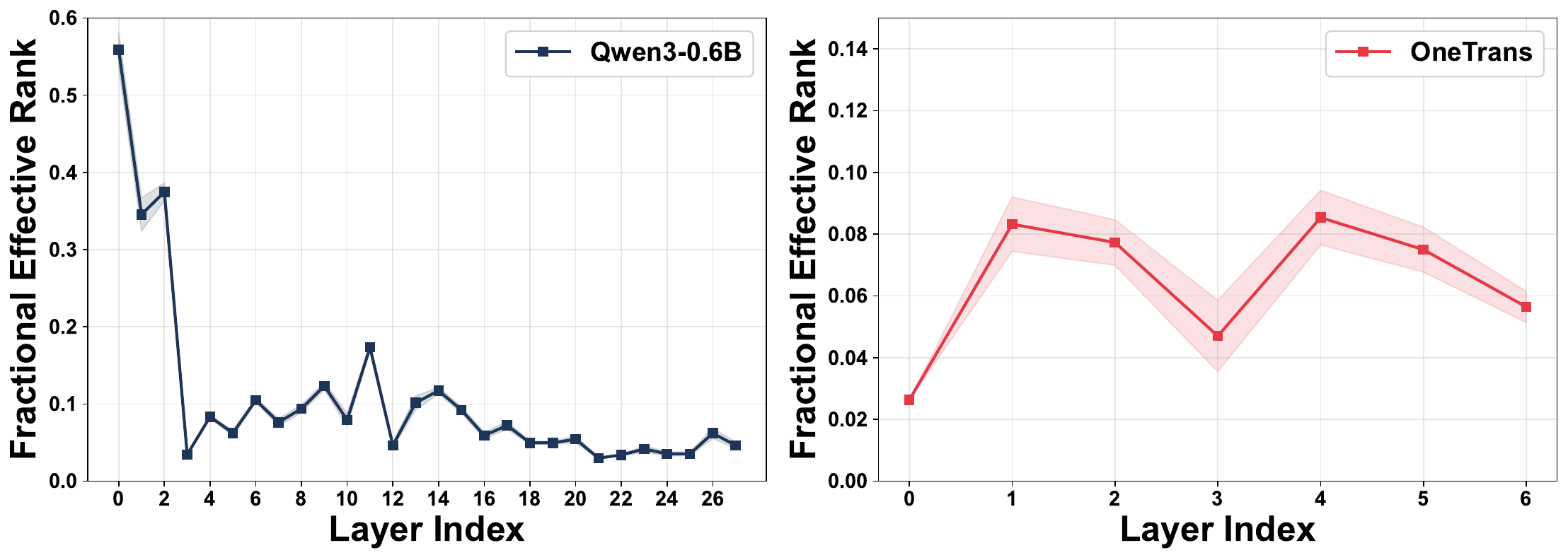}
    \caption{The comparison on fractional effective rank of attention matrix with the increase of Transformer layers. Unlike LLMs, attention matrix in recommendation suffers from serious collapse at the early Transformer layers.}
    \label{fig:erank_onetrans_LLM}
\end{figure}

\begin{figure*}[t] 
    \centering
    \begin{subfigure}{0.62\linewidth}
        \centering
        \includegraphics[width=\linewidth]{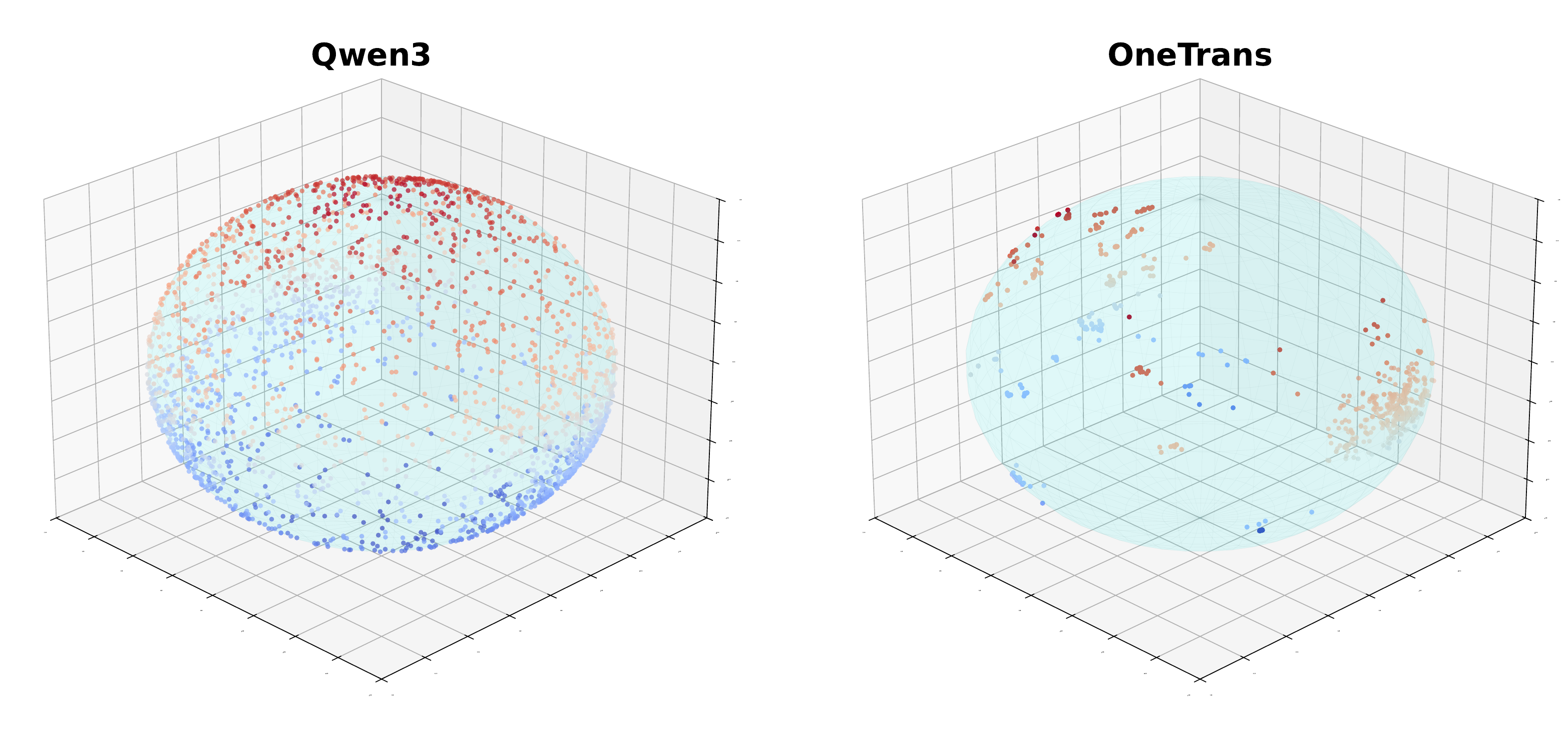}
        \caption{} 
        \label{fig:motivation_heterogeneity_a}
    \end{subfigure}
    \hfill
    \begin{subfigure}{0.28\linewidth}
        \centering
        \includegraphics[width=\linewidth]{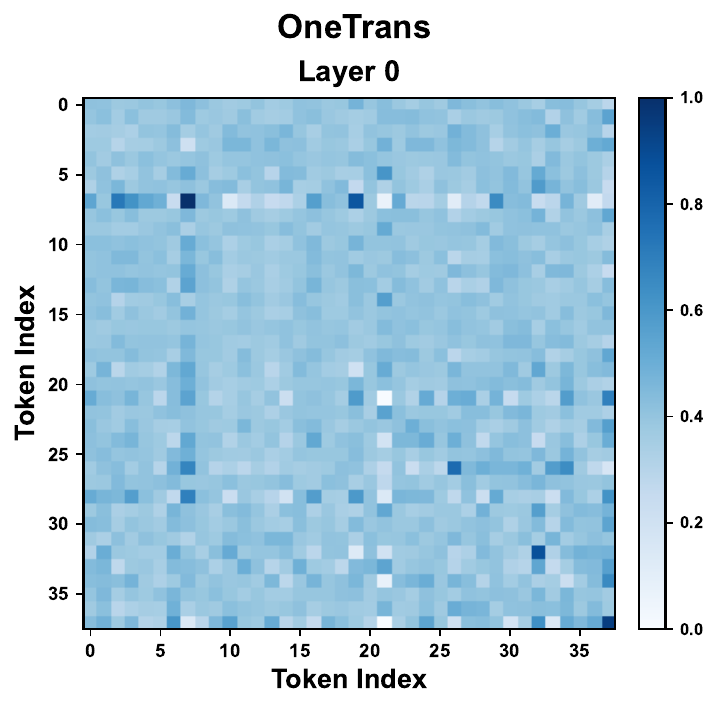}
        \caption{} 
        \label{fig:motivation_heterogeneity_b}
    \end{subfigure}
    
    \caption{(a) Visualization of the token semantic space distributions of Qwen3 and OneTrans, indicating evident heterogeneity in the recommendation semantic space; (b) Heatmap of the first-layer attention score matrix of OneTrans, displaying a highly uniform distribution. A similar phenomenon is observed across various data samples.}
    \label{fig:motivation_heterogeneity}
\end{figure*}

\begin{figure}[t] 
    \centering
    \begin{subfigure}{0.50\linewidth}
        \centering
        \includegraphics[width=\linewidth]{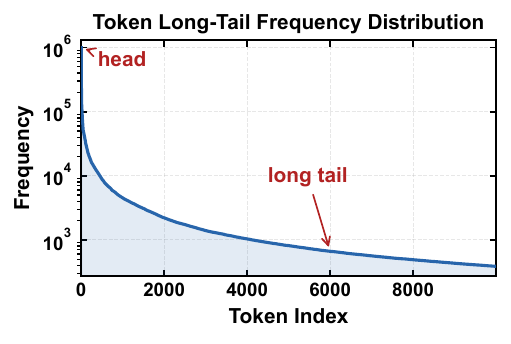}
        \caption{} 
        \label{fig:motivation_long-tail_a}
    \end{subfigure}
    \hfill
    \begin{subfigure}{0.46\linewidth}
        \centering
        \includegraphics[width=\linewidth]{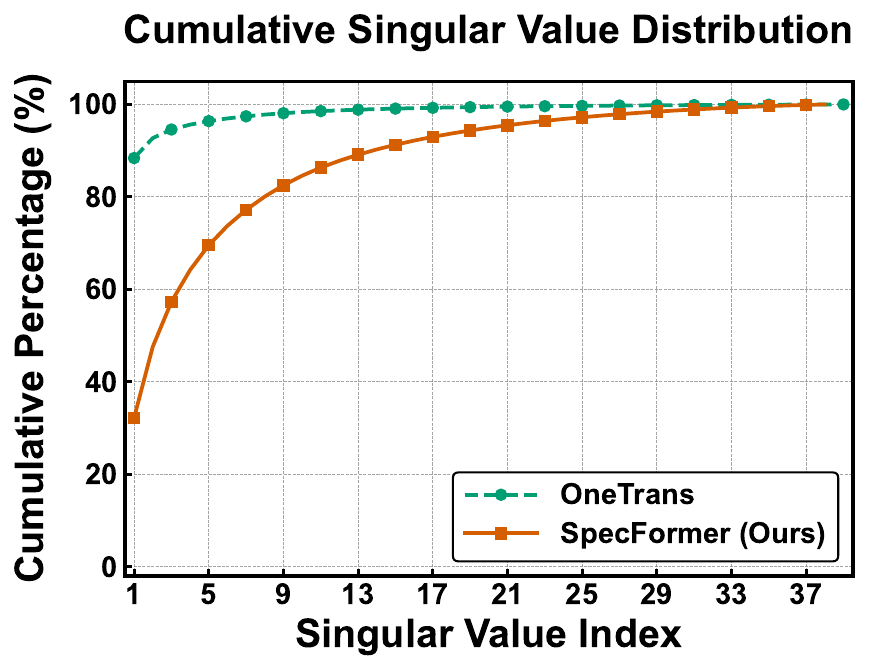}
        \caption{} 
        \label{fig:motivation_long-tail_b}
    \end{subfigure}
    
    \caption{(a) Token frequency exhibits a severe long-tail distribution in recommendation. (b) Cumulative  singular value distribution of OneTrans and our method. OneTrans suffers from severe embedding collapse where top singular values dominate, and SpecFormer effectively mitigates this collapse.}
    \label{fig:motivation_long_tail}
\end{figure}

Transformer architectures~\cite{vaswani2017attention} have achieved remarkable success across  Natural Language Processing (NLP)~\cite{devlin2019bert,brown2020language} and Computer Vision (CV)~\cite{dosovitskiy2020image,liu2021swin}, largely attributed to the expressive power of the self-attention mechanism in capturing long-range dependencies. This success has motivated extensive efforts to adapt Transformers for Recommender Systems (RS)~\cite{zhang2026onetrans,zhu2023final}. Beyond their initial applications in sequential recommendation~\cite{sun2019bert4rec,kang2018self}, Transformers have increasingly emerged as the core architecture for feature interaction modeling. By regarding diverse feature data as individual tokens, the attention mechanism can extract high-order cross-feature dependencies, which has been widely deployed in large-scale industrial recommendation applications~\cite{song2019autoint,zhu2025rankmixer}.

However, directly applying the standard self-attention mechanism to recommendation scenarios often hits a severe performance bottleneck. 
Our empirical observations (\cf Table~\ref{tab:main}) reveal that the state-of-the-art attention-based feature interaction models (\eg  OneTrans~\cite{zhang2026onetrans}) sometimes even trail behind carefully designed simple feature interaction architectures (\eg  RankMixer~\cite{zhu2025rankmixer}). More critically, unlike LLMs, which benefit significantly from depth~\cite{yang2025qwen3}, Transformers in recommendation suffer from serious attention collapse at the early Transformer layers (\cf Figure~\ref{fig:erank_onetrans_LLM}), which results in poor scaling capabilities: stacking layers often leads to severe performance degradation and a sharp drop in the effective rank of the attention matrix (\cf \Cref{sec:scaling}). This naturally raises a key question: \textbf{\textit{Why does the attention mechanism yield sub-optimal performance and collapse so quickly in recommender systems?}}

Through empirical and theoretical analysis, we reveal that it is the heterogeneity and long-tail distribution characteristics of recommendation data that catalyzes a vicious cycle of spectral collapse for both feature embeddings and attention maps in recommendation:

\textbf{(1) Inherent Data Heterogeneity and Long-Tail Distributions Induce Spectral Collapse.} 
Unlike LLMs, where text tokens share a relatively continuous and smooth semantic space, recommendation features (\eg user profiles, item IDs, and contexts) are highly heterogeneous. Our empirical analysis illustrates that recommendation token embeddings do not distribute uniformly; rather, they form highly segregated, discrete semantic sub-spaces (\cf Figure~\ref{fig:motivation_heterogeneity_a}). Consequently, when attention mechanisms attempt to model these tokens, they fail to establish smooth semantic alignments across these heterogeneous fields, leading to overly sparse attention maps (\cf Figure~\ref{fig:motivation_heterogeneity_b}). Furthermore, recommendation data inherently exhibits a severe long-tail distribution in feature frequencies, which leads to the long-tail distribution of input feature tokens (\cf Figure~\ref{fig:motivation_long-tail_a}). As demonstrated by previous work~\cite{lin2025recommendation}, this triggers a critical \textit{spectral collapse} in the representation space, where the embedding matrix is dominated by a few principal singular values, while the minor spectral components are essentially obliterated (\cf Figure~\ref{fig:motivation_long-tail_b}).

\textbf{(2) The Vicious Cycle of Collapse in Forward and Backward Propagation.} 
Most alarmingly, we theoretically demonstrate that this spectral collapse is not merely a static flaw, but triggers a vicious cycle within the Transformer architecture itself (\cf \Cref{sec:theory}).
In the forward propagation, due to the collapsed input token embeddings, the computed attention matrix degenerates into a sparse, low-rank form dominated almost entirely by the principal spectral components. As layers stack, this low-rank attention acts as a low-pass filter, aggressively smoothing out token representations.
During backward propagation, gradients flowing through this low-rank attention matrix are forcefully projected onto the dominant principal sub-spaces. Consequently, minor spectral directions suffer from gradient starvation, receiving almost no effective updates. This biased update rule further aggravates the spectral collapse in the next iteration.
This vicious cycle explains why Transformers in recommendation collapse at much shallower layers compared to the LLMs (\cf Figure~\ref{fig:erank_onetrans_LLM}). 

While several existing efforts have explored various techniques to mitigate the spectral collapse in recommendation, they still face critical limitations. Specifically, some methods attempt to introduce multi-embedding paradigms~\cite{guo2023embedding,chen2026rankup}; however, this introduces prohibitive parameter overhead, making it highly impractical for large-scale systems. Other approaches introduce extensive FFN parameters~\cite{zhu2025rankmixer} or external knowledge integration~\cite{yan2025scaling} to avoid embedding collapse but merely operate in the original spatial domain. However, they fail to fundamentally resolve the root cause of spectral collapse, which yield sub-optimal results (\cf Table~\ref{tab:main} for empirical validation). While a few recent works have explored spectral regularizations~\cite{zhang2024id} or Fourier transforms~\cite{dai2026fedin}, they are either unsuitable for feature interaction modeling or fail to improve the attention collapse. Consequently, exploring a novel Transformer architecture that intrinsically resolves the collapse bottleneck directly from a spectral perspective is of paramount importance.

To tackle these fundamental challenges, we propose \textbf{SpecFormer}, a novel \underline{\textbf{Spec}}tral-Aware Trans\underline{\textbf{former}} architecture specifically tailored for recommendation. SpecFormer breaks the vicious cycle by operating feature interactions explicitly within a dynamically softened spectral domain. Specifically, 1) we first introduce a \textbf{\textit{Learnable Spectral Softening}} module that dynamically re-weights the singular values of the input representations, yielding a more uniform spectral distribution that suppresses dominant components and revitalizes tail features. 2) Building upon this, we propose a \textbf{\textit{Spectrum-softened Attention}} mechanism that computes attention weights within this softened spectral space, forcing the model to capture spectral-aware feature interaction signals. 3) Finally, to provide explicit spectral inductive biases, we design a \textbf{\textit{Spectral Residual Position Encoding}} derived from the Taylor expansion-based mapping function of singular values for attention bias.

Extensive experiments on one industrial and two public benchmark datasets demonstrate the superiority of SpecFormer. It not only significantly outperforms state-of-the-art baselines but also exhibits exceptional scaling capabilities: unlike standard Transformers, stacking layers in SpecFormer actively improves both recommendation performance and the effective rank of attention matrices. \textbf{Notably, SpecFormer has been successfully deployed in a real-world e-commerce advertising platform. In a large-scale online A/B test, it yielded tremendous business gains: a $+1.34\%$ increase in CTR and a $+16.72\%$ boost in orders at an acceptable cost of only a $5ms$ latency overhead.}

Our main contributions are summarized as follows:
\begin{itemize}[leftmargin=*]
    \item We systematically uncover the root cause of Transformer's attention mechanism failure in recommendation. We theoretically and empirically demonstrate the vicious cycle of attention and embedding collapse caused by data heterogeneity and long-tail distributions.
    \item We propose SpecFormer, a spectral-aware Transformer that resolves the collapse bottleneck intrinsically from the attention mechanism via Learnable Spectral Softening, Spectrum-softened Attention, and Spectral Residual Position Encoding.
    \item Extensive offline experiments and online A/B testing verify that SpecFormer achieves state-of-the-art performance. Crucially, it actively increases the attention's effective rank and unlocks the depth-scaling potential of Transformers for recommendation.
\end{itemize}

\section{Preliminary}
\label{sec:Preliminary}
\begin{table}[t]
    \centering
    \caption{Notations in this paper.}
    \label{tab:notation}
    \begin{tabular}{cl}
        \toprule
        \textbf{Notation} & \textbf{Description} \\
        \midrule
        $\mathcal{F}$ & Fields of recommendation data \\
        $L$ & Number of  feature tokens \\
        $D$ & Token embedding dimension size\\
        $\mathbf{X}$ & Initial input feature token matrix \\
        $\mathbf{H}^{(l)}$ & Hidden token representation at layer $l$ \\
        $\mathbf{Q}, \mathbf{K}, \mathbf{V}_{\mathrm{val}}$ & Query, Key, and Value matrices in attention\\
        $\mathbf{W}_q, \mathbf{W}_k, \mathbf{W}_v$ & Projection matrices for Query, Key and Value \\
        $\mathbf{O}$ & The output matrix of self-attention \\
        $\mathbf{S}$ & Pre-softmax attention score matrix \\
        $r$ & Rank of the matrix\\
        $\mathbf{U}$ & Left singular matrix after SVD \\
        $\mathbf{V}$ & Right singular matrix after SVD \\
        $\mathbf{\Sigma}$ & Singular value matrix \\
        $\sigma_i$ & The $i$-th singular value \\
        ${\mathbf{\Sigma}}^*$ & Softened singular matrix \\
        $\tau, a$ & Softening factor and its learnable parameter \\
        ${\mathbf{H}}^*$ & Softened token feature matrix \\
        $\mathbf{P}_{\mathrm{bias}}$ & Learnable spectral residual position encoding bias \\
        $\mathbf{S}_{\mathrm{bias}}$ & Original-space residual bias \\
        $\beta_p$ & Taylor expansion coefficient for the $p$-th term\\
        $\alpha, \gamma$ & Learnable parameters for residual bias strength \\
        \bottomrule
    \end{tabular}
\end{table}

\subsection{Task Formulation}

This paper focuses on Click-Through Rate (CTR) prediction, a crucial and fundamental task in recommender systems~\cite{liu2015convolutional,huang2019fibinet,pan2018field, sun2021fm2}. It aims to estimate the probability that a user will click on a given item, based on a diverse set of categorical and numerical input features~\cite{guo2017deepfm,zhang2023fibinet++}. 

Let $\mathcal{F} = \{f_1, f_2, \ldots, f_F\}$ denote a set of $F$ raw input fields (\emph{e.g.}, user ID, age, item ID, temporal sequences). The CTR prediction task can be formulated as learning a predictive model $f_\theta$:
\begin{equation}
    \hat{y} = \sigma(f_\theta(\mathcal{F})),
\end{equation}
where $\hat{y}$ is the predicted click probability and $\sigma(\cdot)$ is the sigmoid function. The core of $f_\theta$ lies in capturing effective feature interactions. Over the years, a variety of architectures have been proposed for $f_\theta$, ranging from traditional FM-based models~\cite{rendle2010factorization, xiao2017attentional} and DNN-based models~\cite{cheng2016wide, guo2017deepfm} to modern Transformer-based models~\cite{song2019autoint, gui2023hiformer}. Among these, Transformer-based architectures have emerged as the mainstream paradigm due to their exceptional capability in modeling complex, high-order feature interactions~\cite{zhu2025rankmixer,zhang2026onetrans}.

\subsection{Transformer-Based Feature Interaction Model}

A standard Transformer-based feature interaction model typically consists of three sequential stages: input feature tokenization, feature interaction network, and feed-forward network for final prediction.

\textbf{1) Input Feature Tokenization.} 
Traditionally, each raw feature field is mapped to a distinct embedding and fed directly into the model~\cite{song2019autoint}. However, treating hundreds of raw features as individual inputs inevitably leads to parameter fragmentation, insufficient modeling of important features, and under-utilization of GPU computation. To address this, modern Transformer-based models adopt a feature tokenization approach~\cite{zhang2026onetrans,zhu2025rankmixer}. Specifically, the features are first grouped into semantically coherent clusters. The embeddings of these grouped features are sequentially concatenated into a single unified vector $\mathbf{e}_{\mathrm{input}} = [\mathbf{e}_1; \mathbf{e}_2; \ldots; \mathbf{e}_L]$. Subsequently, $\mathbf{e}_{\mathrm{input}}$ is partitioned and projected into an appropriate number of tokens with fixed dimensions:
\begin{equation}
    \mathbf{x}_i = \mathrm{Proj}\left(\mathbf{e}_{\mathrm{input}}[d \cdot (i - 1) : d \cdot i]\right), \quad i = 1, \ldots, L,
\end{equation}
where $d$ is the dimension size of each partitioned chunk, $L$ is the resulting token count, and the $\mathrm{Proj}(\cdot)$ is the MLP that maps the split embedding into a unified dimension $D$. The tokenized feature matrix is then defined as $\mathbf{X} = [\mathbf{x}_1; \mathbf{x}_2; \ldots; \mathbf{x}_L] \in \mathbb{R}^{L \times D}$.

\textbf{2) Feature Interaction Network.} 
Taking the tokenized matrix $\mathbf{H}^{(0)} = \mathbf{X}$ as input, the model relies on feature interaction networks to model user preferences. The representation is updated via a residual connection and Layer Normalization (LN)~\cite{ba2016layer}:
\begin{equation}
    \mathbf{Z}^{(l)} = \mathrm{LN}\!\left(\mathrm{FeatureInteraction}\!\left(\mathbf{H}^{(l)}\right) + \mathbf{H}^{(l)}\right).
\end{equation}
The $\mathrm{FeatureInteraction}(\cdot)$ module explicitly models the cross-feature relations, and there are various designs for this component. A common approach is utilizing the standard self-attention mechanism (\emph{e.g.}, OneTrans~\cite{zhang2026onetrans}). In contrast to standard attention, some recent works introduce carefully designed simplified mechanisms (\emph{e.g.}, RankMixer~\cite{zhu2025rankmixer}) to achieve efficient feature interaction without heavy projection matrices.

\textbf{3) Feed-Forward Network.} 
Following the interaction module, a Feed-Forward Network (FFN) is applied to further transform the token representations. While traditional Transformers share FFN parameters across all tokens, recent mainstream Transformer-based models often adopt a Per-token FFN (PFFN)~\cite{zhu2025rankmixer} as an advanced alternative. By applying dedicated transformations to each token, PFFN preserves the diversity of disparate semantic sub-spaces and prevents high-frequency features from drowning out long-tail signals. The token updating process is formulated as:
\begin{equation}
    \mathbf{H}^{(l+1)} = \mathrm{LN}\!\left(\mathbf{Z}^{(l)} + \mathrm{PFFN}\!\left(\mathbf{Z}^{(l)}\right)\right).
\end{equation}
After propagating through several layers, the final representation is aggregated (\emph{e.g.}, via mean pooling) and fed into a prediction head to yield the final predicted CTR $\hat{y}$.

\subsection{Embedding Collapse and Attention Collapse}
\label{sec:collapse_def}

To formally quantify the collapse phenomena in Transformer-based recommendation, we first introduce the Singular Value Decomposition (SVD)~\cite{hoecker1996svd}. For a matrix $\mathbf{H} \in \mathbb{R}^{L \times D}$ with rank $r = \min(L, D)$, the SVD can be expressed as:
\begin{equation}
\mathbf{H} = \mathbf{U} \mathbf{\Sigma} \mathbf{V}^T,
\end{equation}
where $\mathbf{U} \in \mathbb{R}^{L \times r}$ and $\mathbf{V} \in \mathbb{R}^{D \times r}$ are orthonormal matrices whose columns are the left and right singular vectors, respectively, and $\mathbf{\Sigma} = \mathrm{diag}(\sigma_1, \sigma_2, \ldots, \sigma_r)$ with $\sigma_1 \geq \sigma_2 \geq \cdots \geq \sigma_r \geq 0$ is the singular value matrix. 

\textbf{1) Embedding Collapse.} 
Embedding collapse refers to the phenomenon where the input token representations degenerate into a narrow, low-dimensional subspace~\cite{chen2024towards}. To evaluate this, we apply SVD on the token embedding matrix $\mathbf{H} \in \mathbb{R}^{L \times D}$. In the typical CTR setting, $L < D$, meaning the maximum possible rank is $r = L$. Following previous work~\cite{cui2026spectran}, we use the \textit{cumulative singular value ratio} to quantify the degree of embedding dimension collapse:
\begin{equation}
    \rho(k) = \frac{\sum_{i=1}^{k} \sigma_i^2}{\sum_{i=1}^{r} \sigma_i^2}, \quad k = 1, \ldots, r.
\end{equation}
A rapid saturation of $\rho(k)$ toward almost $100\%$ at a small singular value index $k$ indicates severe embedding collapse, meaning the feature space is dominated by only a few principal components.

\textbf{2) Attention Collapse.} 
Attention collapse occurs when the attention map loses its discriminative capability, degenerating into a smoothed or overly sparse pattern. To measure this, we apply SVD on the attention score matrix $\mathbf{S} \in \mathbb{R}^{L \times L}$ (where $r=L$). We quantify this degeneration using the \textit{Effective Rank}~\cite{roy2007effective}, which calculates the Shannon entropy~\cite{shannon1948mathematical} of the singular value distribution. Effective Rank is also widely used in LLMs to evaluate attention collapse~\cite{vaswani2017attention,noci2022signal,wei2024diff}. First, we compute the $L_1$-normalized singular values $p_i$ as:
\begin{equation}
    p_i = \frac{\sigma_i}{\sum_{j=1}^{r} \sigma_j}, \quad i = 1, \ldots, r,
\end{equation}
where $\sigma_i$ is the $i$-th singular value of $\mathbf{S}$. Then, the Effective Rank (erank) is defined as the exponential of the Shannon entropy of this distribution:
\begin{equation}
    \mathrm{erank}(\mathbf{S}) = \exp \left( - \sum_{i=1}^{r} p_i \ln p_i \right).
\end{equation}
The value of $\mathrm{erank}(\mathbf{S})$ ranges continuously from $1$ (a completely collapsed rank-1 matrix) to $r$ (a perfectly isotropic full-rank matrix).  Furthermore, to provide a scale-invariant metric that facilitates fair comparisons across different sequence lengths or model dimensions, we use the \textit{Fractional Effective Rank}~\cite{wang2021dcn,chen2026rankup} (ferank) by normalizing the Effective Rank by the maximum possible rank $r$:
\begin{equation}
    \mathrm{ferank}(\mathbf{S}) = \frac{\mathrm{erank}(\mathbf{S})}{r}.
\end{equation}
The $\mathrm{ferank}(\mathbf{S})$ inherently bounds the metric within the range $[0, 1]$. A fractional effective rank approaching $0$ serves as a standardized indicator of severe attention collapse, whereas a value of $1$ signifies an optimal, perfectly full-rank attention map regardless of the absolute matrix size.

\section{Theoretical Analysis}
\label{sec:theory}

In contrast to LLMs where token embeddings typically exhibit rich semantic variance, recommendation systems heavily rely on sparse recommendation data with severe power-law distributions (\cf Figure~\ref{fig:motivation_heterogeneity} and ~\ref{fig:motivation_long_tail}). This unique characteristic often leads to severe spectral collapse right at the input embedding layer (\cf Figure~\ref{fig:erank_onetrans_LLM}). In this section, we formally demonstrate that this inherent embedding collapse in recommendation scenarios triggers a vicious cycle of degradation through both the forward and backward propagation of standard attention mechanism. For mathematical simplicity, we omit the standard scaling factor $1/\sqrt{D}$.

\begin{proposition}[Forward Propagation: Severe Recommendation Embedding Collapse Induces Rapid Attention Collapse]
\label{prop:forward_collapse}
Let the input feature matrix be $\mathbf{H} \in \mathbb{R}^{L \times D}$ with its SVD denoted as $\mathbf{H} = \mathbf{U\Sigma V}^\top$, where $\mathbf{U}$ and $\mathbf{V}$ are orthogonal matrices, and $\mathbf{\Sigma} = \mathrm{diag}(\sigma_1, \ldots, \sigma_r)$. Given the severe spectral collapse premise in recommendation scenarios, assume $\sigma_k \leq \epsilon \sigma_1$ for all $k \geq 2$ ($0 < \epsilon \ll 1$). Then, the attention map degenerates and is almost entirely dominated by the principal spectral component, losing feature discriminability.
\end{proposition}

\begin{proof}
The pre-softmax attention score matrix is computed as $\mathbf{S} = \mathbf{H}\mathbf{W}_q\mathbf{W}_k^\top\mathbf{H}^\top$. By substituting the SVD of $\mathbf{H}$ and defining a projection core matrix $\mathbf{M} = \mathbf{V}^\top \mathbf{W}_q\mathbf{W}_k^\top \mathbf{V} \in \mathbb{R}^{r \times r}$, we can express $\mathbf{S}$ as:
\begin{equation}
    \mathbf{S} = \mathbf{U} \mathbf{\Sigma} (\mathbf{V}^\top \mathbf{W}_q\mathbf{W}_k^\top \mathbf{V}) \mathbf{\Sigma} \mathbf{U}^\top = \mathbf{U} \mathbf{\Sigma} \mathbf{M} \mathbf{\Sigma} \mathbf{U}^\top.
\end{equation}
The $(i,j)$-th entry of $\mathbf{S}$ can be explicitly expanded and decomposed into:
\begin{equation}
    \begin{aligned}
    S_{ij} &= \sum_{a=1}^r \sum_{b=1}^r u_{ia} \sigma_a M_{ab} \sigma_b u_{jb} \\
           &= \underbrace{\sigma_1^2 M_{11} u_{i1} u_{j1}}_{\text{Dominant Component}} + \underbrace{\sum_{(a,b) \neq (1,1)} \sigma_a \sigma_b M_{ab} u_{ia} u_{jb}}_{\text{Minor Component} R_{ij}}.
    \end{aligned}
\end{equation}
Since $\mathbf{U}$ and $\mathbf{V}$ are orthogonal matrices ($|u_{ij}| \le 1$), and the learned weight matrices $\mathbf{W}_q, \mathbf{W}_k$ inherently have bounded norms due to standard regularization, their product $\mathbf{M} = \mathbf{V}^\top \mathbf{W}_q\mathbf{W}_k^\top \mathbf{V}$ is bounded at a relatively small scale. Consequently, the magnitude of each term is strictly governed by the singular values $\sigma_a \sigma_b$. Under the embedding collapse condition ($\sigma_k \leq \epsilon \sigma_1$ for $k \geq 2$), the residual term is bounded by $|R_{ij}| \leq \mathcal{O}(\epsilon \sigma_1^2)$. Therefore, the pre-softmax score is heavily dominated by the rank-1 principal component. Worse still, when passed through the exponential function $\mathbf{A} = \mathrm{softmax}(\mathbf{S})$, this dominance is exponentially amplified. The attention matrix $\mathbf{A}$ loses sensitivity to minor feature interactions and rapidly collapses into a low-rank matrix governed primarily by $\mathbf{u}_1$.
\end{proof}

\begin{proposition}[Backward Propagation: Gradient Concentration on Dominant Components Amplifies Collapse]
\label{prop:backward_collapse}
Given the smoothed and degenerate attention matrix $\mathbf{A}$ derived from Proposition~\ref{prop:forward_collapse}, the gradients during backpropagation are heavily concentrated on the dominant spectral component, while minor components receive negligible updates, rendering the model incapable of recovering lost feature diversity.
\end{proposition}

\begin{proof}
Consider the backward propagation through the value aggregation phase $\mathbf{O} = \mathbf{A}\mathbf{X}$ (where $\mathbf{X} = \mathbf{H}\mathbf{W}_v$). (While gradients also backpropagate through the query and key projections, those updates inherently involve right-multiplications with the highly rank-deficient input $\mathbf{H}$, meaning their gradients are also strictly confined to the dominant spectral subspace~\cite{dong2021attention}). Focusing on the gradient path through the value aggregation, by the chain rule, the gradient with respect to the input representations $\mathbf{X}$ is:
\begin{equation}
    \nabla_{\mathbf{X}} \mathcal{L} = \mathbf{A}^\top \nabla_{\mathbf{O}} \mathcal{L}.
\end{equation}
From Proposition~\ref{prop:forward_collapse}, the pre-softmax score matrix can be expressed as $\mathbf{S} = \mathbf{S}^{1} + \mathbf{R}$, where $\mathbf{S}^{1} = \sigma_1^2 M_{11} \mathbf{u}_1 \mathbf{u}_1^\top$ is the rank-1 dominant component and the residual satisfies $\|\mathbf{R}\|_\infty = \mathcal{O}(\epsilon \sigma_1^2)$. Let $\mathbf{A}^{1} = \mathrm{softmax}(\mathbf{S}^{1})$. Due to the Lipschitz continuity of the softmax function~\cite{gao2017properties}, the actual attention matrix satisfies $\mathbf{A} = \mathbf{A}^{1} + \mathbf{A}^{E}$, where the error matrix $\mathbf{A}^{E}$ is bounded by $\|\mathbf{A}^{E}\| \leq \mathcal{O}(\epsilon\sigma_1^2)$.

To see how this affects the gradient flow, we project the upstream gradient $\nabla_{\mathbf{O}} \mathcal{L}$ onto the orthogonal basis $\mathbf{U}$:
\begin{equation}
    \nabla_{\mathbf{O}} \mathcal{L} = \mathbf{u}_1 \mathbf{g}_1^\top + \sum_{k=2}^r \mathbf{u}_k \mathbf{g}_k^\top,
\end{equation}
where $\mathbf{g}_k$ represents the gradient vector corresponding to the $k$-th spectral component. The backward update can then be expanded as:
\begin{equation}
    \begin{aligned}
    \nabla_{\mathbf{X}} \mathcal{L} &= (\mathbf{A}^{1} + \mathbf{A}^{E})^\top \left( \mathbf{u}_1 \mathbf{g}_1^\top + \sum_{k=2}^r \mathbf{u}_k \mathbf{g}_k^\top \right).
    \end{aligned}
\end{equation}
As proven by the previous work~\cite{dong2021attention}, the self-attention matrix acts intrinsically as a low-pass filter. Because $\mathbf{A}^{1}$ is generated by the latent structure of $\mathbf{u}_1$, its rows exhibit minimal variance and lack high-frequency complexity. Consequently, when operating on the orthogonal basis $\mathbf{U}$, the transformation $\mathbf{A}^{1\top}$ strongly preserves the primary structural component $\mathbf{u}_1$, but severely attenuates the high-frequency minor components $\mathbf{u}_k$ ($k \geq 2$). 

Let $\xi_k = \|\mathbf{A}^{1\top} \mathbf{u}_k\| / \|\mathbf{u}_k\|$ denote the attenuation factor, which quantifies the fraction of the $k$-th spectral component's magnitude preserved during backpropagation. For minor spectral components, the low-pass filtering property dictates that $\xi_k \ll \xi_1$. Thus, the gradient update projected onto the $k$-th minor component ($k \geq 2$) is bounded by:
\begin{equation}
    \| \nabla_{\mathbf{X}} \mathcal{L} \cdot \mathbf{u}_k \| \leq \xi_k \|\mathbf{g}_k\| + \mathcal{O}(\epsilon\sigma_1^2).
\end{equation}
Since the attenuation factor $\xi_k$ is exceedingly small and the residual perturbation is bounded by $\mathcal{O}(\epsilon\sigma_1^2)$, the gradient is overwhelmingly concentrated on the dominant component $\mathbf{u}_1$. The minor components receive highly attenuated updates, which exacerbates the embedding collapse and prevents the model from escaping this degenerate state.
\end{proof}

\textbf{The Vicious Cycle of Collapse in Recommendation.}
Propositions~\ref{prop:forward_collapse} and~\ref{prop:backward_collapse} collectively reveal a closed-loop vicious cycle: initially collapsed embeddings in recommendation generate smoothed, low-rank attention maps (Forward Propagation). These degenerate attention maps, acting as low-pass filters, strictly restrict the gradient updates to the dominant spectral directions (Backward Propagation). This pushes the embeddings to become even more collapsed in the next iteration. This fundamental mechanism mathematically explains why standard self-attention suffers from attention collapse much faster and more severely in recommendation scenarios compared to LLMs.

\section{Methodology}
\label{sec:method}

\begin{figure*}[t]
    \centering
    \includegraphics[width=\linewidth]{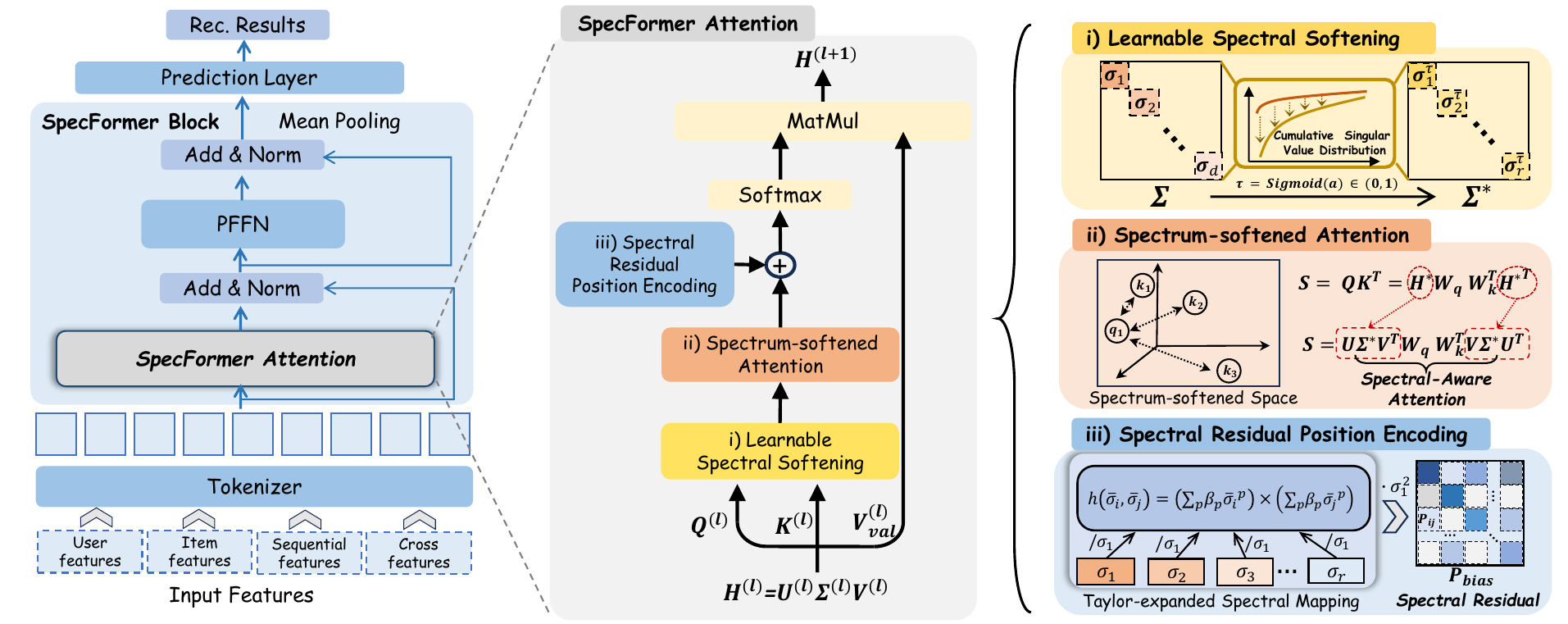}
    \caption{The architecture of  SpecFormer. Each SpecFormer attention layer consists of three components: (1) Learnable Spectral Softening applied to the input feature matrix; (2) Spectrum-softened Attention operating in the spectral-aware feature interaction space; and (3) Spectral Residual Position Encoding providing a learnable Taylor-expanded spectral bias term.}
    \label{fig:overview}
\end{figure*}

To address the embedding collapse and attention collapse in recommendation, we propose \textbf{SpecFormer}, a Spectral-Aware Transformer tailored for recommendation. As illustrated in Figure~\ref{fig:overview}, each SpecFormer layer replaces the standard self-attention module with three tightly coupled spectral components while keeping the remaining Transformer components unchanged. We describe each component in detail below.

\subsection{Learnable Spectral Softening}
\label{sec:softening}

The root cause of embedding collapse is the skewed singular value distribution of the feature embedding matrix, where a few dominant singular values capture nearly all the representational information. We directly address this by applying a learnable spectral softening to the input at each layer before attention computation.

Formally, let $\mathbf{H}^{(l)} \in \mathbb{R}^{L \times D}$ denote the input to the $l$-th Transformer layer, with $\mathbf{H}^{(0)} = \mathbf{X}$. We compute the SVD as:
\begin{equation}
    \mathbf{H}^{(l)} = \mathbf{U}^{(l)} \mathbf{\Sigma}^{(l)} {\mathbf{V}^{(l)}}^\top,
\end{equation}
where $\mathbf{U}^{(l)} \in \mathbb{R}^{L \times r}$, $\mathbf{V}^{(l)} \in \mathbb{R}^{D \times r}$, $r = \min(L, D)$, and $\mathbf{\Sigma}^{(l)} = \mathrm{diag}(\sigma_1^{(l)}, \ldots, \sigma_r^{(l)})$ with $\sigma_1^{(l)} \geq \cdots \geq \sigma_r^{(l)} \geq 0$.

We then apply a \textit{power-law softening} to the singular value spectrum:
\begin{equation}{\mathbf{\Sigma}^{(l)}}^* = \mathrm{diag}\!\left({\sigma_1^{(l)}}^{\tau},\, {\sigma_2^{(l)}}^{\tau},\, \ldots,\, {\sigma_r^{(l)}}^{\tau}\right),\label{eq:softening}
\end{equation}
where $\tau = \mathrm{Sigmoid}(a) \in (0, 1)$ and $a$ is a learnable scalar initialized to $0$ (\ie $\tau$ starts at $0.5$). Since $\tau < 1$, the power-law transformation compresses large singular values more aggressively than small ones, effectively flattening the spectrum and redistributing representational energy across all dimensions. The softened feature matrix is then reconstructed as:
\begin{equation}{\mathbf{H}^{(l)}}^* = \mathbf{U}^{(l)} {\mathbf{\Sigma}^{(l)}}^* {\mathbf{V}^{(l)}}^\top.
    \label{eq:softened_H}
\end{equation}

The power-law softening provides a continuous, data-adaptive control over the degree of spectral flattening via the learnable parameter $\tau$.

\subsection{Spectrum-softened Attention}
\label{sec:spec_attn}

Standard self-attention computes feature interactions directly in the original feature space, which causes attention to be dominated by the collapsed principal component. We instead propose to compute attention weights in the \textit{spectrum-softened domain} with the softened embeddings, so that all spectral components contribute meaningfully to feature interactions.

Specifically, we define the Query and Key matrices using the softened embedding ${\mathbf{H}^{(l)}}^*$, while keeping the Value matrix in the original feature space:
\begin{equation}
    \mathbf{Q}^{(l)} = {\mathbf{H}^{(l)}}^* \mathbf{W}_q^{(l)}, \quad
    \mathbf{K}^{(l)} = {\mathbf{H}^{(l)}}^* \mathbf{W}_k^{(l)}, \quad
    \mathbf{V}_{\mathrm{val}}^{(l)} = \mathbf{H}^{(l)},
    \label{eq:qkv}
\end{equation}
where $\mathbf{W}_q^{(l)}, \mathbf{W}_k^{(l)} \in \mathbb{R}^{D \times D}$. Using the original $\mathbf{H}^{(l)}$ (without softening or projection) as the Value matrix is motivated empirically: the Value matrix aggregates feature content, and preserving the original representation avoids information loss in the output. We also compare it with other Value matrix designs and find it performs best (\cf Section~\ref{sec:ablation}).

The pre-softmax attention score matrix is:
\begin{equation}
    \mathbf{S}^{(l)} = \mathbf{Q}^{(l)}{\mathbf{K}^{(l)}}^\top = {\mathbf{H}^{(l)}}^* \mathbf{W}_q^{(l)} {\mathbf{W}_k^{(l)}}^\top {\mathbf{H}^{(l)}}^{*\top}.
    \label{eq:score}
\end{equation}

To reveal the spectral nature of this computation, we substitute Eq.~\eqref{eq:softened_H} into Eq.~\eqref{eq:score} and drop the layer superscript for clarity:
\begin{equation}
    \mathbf{S} = \mathbf{U}\underbrace{\left(\mathbf{\Sigma}^* \mathbf{V}^\top \mathbf{W}_q \mathbf{W}_k^\top \mathbf{V} \mathbf{\Sigma}^*\right)}_{\text{spectrum-softened attention} }\mathbf{U}^\top.
    \label{eq:spectral_score}
\end{equation}
This design naturally prevents a single spectral component from dominating the attention computation, as the softened spectrum $\mathbf{\Sigma}^*$ ensures a more balanced contribution from all components.


\subsection{Spectral Residual Position Encoding}
\label{sec:spec_residual}

In addition to the spectrum-softened attention, we further introduce a \textit{learnable spectral residual position encoding} to provide an explicit spectral inductive bias for feature interactions and prevent collapse from compounding across layers.

\textbf{Spectral Residual Term.}
We define a spectral position encoding bias $\mathbf{P}_{\mathrm{bias}}^{(l)} \in \mathbb{R}^{L \times L}$ as:
\begin{equation}
    \mathbf{P}_{\mathrm{bias}}^{(l)} = \mathbf{U}^{(l)} \mathbf{P}^{(l)} {\mathbf{U}^{(l)}}^\top,
    \label{eq:spec_bias}
\end{equation}
where $\mathbf{P}^{(l)} \in \mathbb{R}^{r \times r}$ is a learnable spectral interaction weight matrix. To parameterize $\mathbf{P}^{(l)}$ in a structured and data-adaptive manner, we define its $(i,j)$-th entry via a \textit{Taylor-expanded pairwise spectral mapping}:
\begin{align}
    \bar{\sigma}_i &= \sigma_i / \sigma_1, \label{eq:norm_sv} \\
    h(\bar{\sigma}_i, \bar{\sigma}_j) &= \left(\sum_{p=0}^{2} \beta_p \bar{\sigma}_i^p\right)\times\left(\sum_{p=0}^{2} \beta_p \bar{\sigma}_j^p\right), \label{eq:taylor} \\
    P_{ij}^{(l)} &= {\sigma}_1^2\, h(\bar{\sigma}_i, \bar{\sigma}_j), \label{eq:A_entry}
\end{align}
where $\bar{\sigma}_i$ is the normalized singular value, $h(\bar{\sigma}_i, \bar{\sigma}_j)$ is a second-order polynomial that models the pairwise spectral interaction between the $i$-th and $j$-th feature components, and $\boldsymbol{\beta} = [\beta_0, \beta_1, \beta_2]^\top$ are shared learnable parameters. 

Crucially, this specific parameterization provides a vital structural inductive bias to compensate for the original spectral weights. Because the spectral softening (Section~\ref{sec:softening}) effectively flattens the representation space to revitalize tail features, it might inadvertently penalize the highly informative dominant components. By explicitly anchoring the interactions on the maximum singular value $\sigma_1^2$ and the normalized spectrum, this bias term acts as a targeted structural residual. It preserves the heavy-weighted principal spectral signals, preventing the loss of main recommendation patterns in the uniformized space.

\textbf{Spatial Residual Term.}
To further stabilize training and preserve the feature interaction signal from the original space, we add a second residual term computed from the un-softened embeddings:
\begin{equation}
    \mathbf{S}_{\mathrm{bias}}^{(l)} = \left( \mathbf{H}^{(l)} \tilde{\mathbf{W}}_q^{(l)} \right) \left( \mathbf{H}^{(l)} \tilde{\mathbf{W}}_k^{(l)} \right)^{\top},
    \label{eq:orig_bias}
\end{equation}

where $\tilde{\mathbf{W}}_q^{(l)}$ and $\tilde{\mathbf{W}}_k^{(l)}$ are independent linear projections dedicated to the residual connection, ensuring the gradient flows of the original and softened spaces are independent.

\textbf{Final SpecFormer Attention.}
Combining the spectral attention score (Eq.~\ref{eq:score}), the spectral bias (Eq.~\ref{eq:spec_bias}), and the spatial residual (Eq.~\ref{eq:orig_bias}), the complete SpecFormer attention is:
\begin{equation}
    \mathrm{SpecAtt}(\mathbf{H}^{(l)}) = \phi\!\left(\mathbf{S}^{(l)} + \alpha\, \mathbf{P}_{\mathrm{bias}}^{(l)} + \gamma\, \mathbf{S}_{\mathrm{bias}}^{(l)}\right) \mathbf{V}_{\mathrm{val}}^{(l)},
    \label{eq:specatt}
\end{equation}
where $\phi(\cdot) = \mathrm{softmax}(\cdot)$ is the activation function, and $\alpha, \gamma \in \mathbb{R}$ are hyperparameters that control the strength of each residual term.

The remaining components follow the standard Transformer design. The full update rule for the $l$-th SpecFormer layer is:
\begin{align}
    \mathbf{Z}^{(l)} &= \mathrm{LN}\!\left(\mathrm{SpecAtt}(\mathbf{H}^{(l)}) + \mathbf{H}^{(l)}\right), \label{eq:ln1} \\
    \mathbf{H}^{(l+1)} &= \mathrm{LN}\!\left(\mathbf{Z}^{(l)} + \mathrm{PFFN}(\mathbf{Z}^{(l)})\right), \label{eq:ln2}
\end{align}
where $\mathrm{LN}(\cdot)$ is Layer Normalization~\cite{ba2016layer} and $\mathrm{PFFN}(\cdot)$ is the pre-token FFN introduced in \Cref{sec:Preliminary}. The final representation is mean-pooled over the feature dimension and passed to a prediction MLP.

\section{Experiments}
\label{sec:exp}

In this section, we conduct extensive experiments to evaluate the proposed SpecFormer. Specifically, we aim to answer the following research questions:
\begin{itemize}[leftmargin=*]
    \item \textbf{RQ1:} How does SpecFormer perform compared to existing state-of-the-art traditional, Transformer-based, and spectral-based recommendation models?
    \item \textbf{RQ2:} How do the key components of SpecFormer contribute to the overall performance?
    \item \textbf{RQ3:} Does SpecFormer mitigate the collapse bottleneck and achieve scaling capability on stacking Transformer layers?
    \item \textbf{RQ4:} What are the impacts of different hyperparameters in SpecFormer?
    \item \textbf{RQ5:} How does SpecFormer perform in a real-world, large-scale online industrial system in terms of system efficiency and business metrics?
\end{itemize}

\subsection{Experimental Setup}
\label{sec:expSetup}
\subsubsection{Datasets}
Following previous work~\cite{guo2023embedding,zhu2025rankmixer}, we evaluate SpecFormer on the following datasets: one industrial dataset and two well-known public benchmarks --- Criteo\footnote{\url{https://www.kaggle.com/c/criteo-display-ad-challenge/data}}, and Avazu\footnote{\url{https://www.kaggle.com/c/avazu-ctr-prediction/data}}.
Table~\ref{tab:datasets} summarizes the statistics of each dataset.
\begin{itemize}
    \item \textbf{Industrial Dataset} is collected from the online display advertising system of a leading international e-commerce platform. It contains approximately 1.2 billion exposure samples with user behavior sequences up to a length of 256, containing 38 feature fields. 
    
    \item \textbf{Criteo} is a well-known public benchmark for CTR prediction. This dataset contains one week of real-world advertisement click data. It features a mix of 13 continuous (numerical) features and 26 categorical features.
    
    \item \textbf{Avazu} is a widely-used benchmark consisting of 10 days of chronologically ordered ad click logs. Each sample contains 23 feature fields, providing a robust environment for testing on rich time-ordered categorical data.

\end{itemize}

\begin{table}[t]
    \centering
    \caption{Statistics of the datasets used in experiments.}
    \label{tab:datasets}
    \begin{tabular}{lcccc}
        \toprule
        \textbf{Dataset} & \textbf{\#Samples} & \textbf{\#Fields} & \textbf{\#Features} & \textbf{\#Positive} \\
        \midrule
        Industrial & 1.2B  & 391 & $>$1B      & 2.7\% \\
        Criteo     & 45M   & 39  & 1.33M    & 26.0\% \\
        Avazu      & 40M   & 23  & 1.54M    & 17.0\% \\
        \bottomrule
    \end{tabular}
\end{table}

\subsubsection{Baselines}
We compare SpecFormer against three categories of baselines:
\begin{itemize}[leftmargin=*]
  \item \textbf{Traditional Feature Interaction Models:} 
  \textbf{DeepFM (arXiv'17)}~\cite{guo2017deepfm} integrates factorization machines and deep neural networks to jointly model both low-order and high-order feature interactions. 
  \textbf{AutoInt (CIKM'19)}~\cite{song2019autoint} utilizes a multi-head self-attention to explicitly map features into a projected space for automated high-order interaction learning. 
  \textbf{MaskNet (arXiv'21)}~\cite{wang2021masknet} introduces an instance-guided mask module to dynamically highlight salient features and capture complex non-linear feature interactions. 
  \textbf{FiBiNet++ (CIKM'23)}~\cite{zhang2023fibinet++} enhances the FiBiNet~\cite{huang2019fibinet} by refining the bilinear interaction process and attention mechanisms for more efficient representation learning. 
  \textbf{GDCN (CIKM'23)}~\cite{wang2023towards} proposes a gated cross network that captures explicit high-order feature interactions while utilizing gating mechanisms to dynamically filter out noise.

  \item \textbf{Transformer-Based Feature Interaction Models:} 
  \textbf{HiFormer (arXiv'23)}~\cite{gui2023hiformer} designs a hierarchical Transformer architecture to effectively capture both fine-grained intra-feature and coarse-grained inter-feature relationships. 
  \textbf{FAT (arXiv'25)}~\cite{yan2025scaling} scales up feature interaction learning by adapting a comprehensive Transformer architecture to model complex, long-range cross-feature dependencies. 
  \textbf{RankMixer (CIKM'25)}~\cite{zhu2025rankmixer} integrates a specialized token-mixing attention mechanism tailored for ranking tasks to efficiently aggregate and transform feature interactions. 
  \textbf{OneTrans (WWW'26)}~\cite{zhang2026onetrans} proposes a Transformer framework to seamlessly integrate various feature interaction paradigms into a single unified attention mechanism.

  \item \textbf{Spectral-Based Recommendation Models:} 
  \textbf{WhitenRec (ICDE'24)}~\cite{zhang2024id} applies a spectral-based whitening transformation to alleviate dimensional collapse. 
  \textbf{FEDIN (arXiv'26)}~\cite{dai2026fedin} leverages Fourier Transform to optimize feature representations in the frequency domain and model global interaction patterns.
\end{itemize}

\subsubsection{Evaluation Metrics}
Following previous work~\cite{zhu2025rankmixer,zhang2026onetrans}, We employ the widely used \textbf{AUC} (Area Under the ROC Curve) and \textbf{Logloss} (Cross-Entropy Loss) as standard evaluation metrics. For the industrial dataset, we additionally report \textbf{GAUC} (Group AUC), which measures user-specific ranking quality and is highly correlated with online performance. Note that a \textbf{0.001} improvement in AUC/GAUC is considered significant and can lead to massive revenue growth in industrial applications~\cite{zhou2018deep,song2019autoint}.

\subsubsection{Implementation Details}
We implement all models in TensorFlow and conduct experiments on servers equipped with NVIDIA L20 GPUs. For the industrial dataset, the token embedding dimension is set to 304, optimized via Adagrad~\cite{duchi2011adaptive} with a batch size of 1024. For public datasets (Criteo, and Avazu), the embedding dimensions are configured as 312, and 176, respectively, using a batch size of 512. The optimal learning rate for each model is determined through a grid search within $\{0.005, 0.01, 0.05, 0.1\}$ for the industrial dataset and a broader range of $\{0.001, 0.002, 0.005, 0.01, 0.05, 0.1\}$ for public datasets. 
Our code is available at {\url{https://github.com/istarryn/SpecFormer}}.

\subsubsection{Training Strategy and Warm-up Protocol} 
To ensure numerical stability and model convergence, we adopt a two-stage training protocol for \text{SpecFormer}, whereas the baseline models follow a standard single-stage procedure: 
(a) \textit{Baseline Configuration:} Following the common practice in CTR prediction, all baseline models are initialized using a standard Gaussian distribution and trained on the full training set (100\%) to optimize the target CTR loss.
(b) \textit{SpecFormer Configuration:} We employ a two-stage training scheme. In the first stage (Embedding Warm-up), 5\% of the training data is utilized to initialize the embeddings by optimizing the \textit{Spatial Residual Term}. In the second stage (Main CTR Training), the model is trained on the remaining 95\% of the data using the complete architecture, including the Spectral Softening module.

The warm-up phase is essential for stabilizing the \textit{Spectral Softening} mechanism. Randomly initialized embeddings often exhibit ill-conditioned singular value distributions dominated by noise, making direct spectral transformations prone to numerical instability. By pre-training with the \textit{Spatial Residual Term}, we project embeddings onto a structured feature manifold. This establishes a stable spectral foundation, enabling robust frequency-domain refinement during the subsequent CTR training stage.

\subsection{Performance Comparison (RQ1)}
\label{sec:main_results}

Table~\ref{tab:main} shows the performance comparison of the proposed SpecFormer against the baseline methods. We observe that:
\begin{table*}[t]
    \centering
    \caption{Performance comparisons of SpecFormer with different representative baselines on three real-world datasets. A higher AUC/GAUC and a lower Logloss indicate better performance. The best and the second performances are \textbf{bolded} and \underline{underlined}, respectively. }
    \label{tab:main}
    \resizebox{\linewidth}{!}{%
    \begin{tabular}{ccccccccc}
        \toprule
        \multirow{2}{*}{\textbf{Category}} & \multirow{2}{*}{\textbf{Model}} 
        & \multicolumn{3}{c}{\textbf{Industrial}} & \multicolumn{2}{c}{\textbf{Criteo}} & \multicolumn{2}{c}{\textbf{Avazu}} \\
        \cmidrule(lr){3-5}\cmidrule(lr){6-7}\cmidrule(lr){8-9}
        & & AUC ↑& GAUC ↑& Logloss ↓& AUC ↑& Logloss ↓& AUC ↑& Logloss ↓\\
        \midrule
        \multirow{5}{*}{Traditional}
        & Wide\&Deep     & 0.7373 & 0.6287 & 0.1144 & 0.6589 & 0.5372 & 0.6364 & 0.4335 \\
        & DeepFM         & 0.7378 & 0.6282 & 0.1143 & 0.5988 & 0.5421 & 0.7404 & 0.3965 \\
        & GDCN           & 0.7491 & 0.6426 & 0.1133 & 0.6727 & 0.5268 & 0.7370 & 0.3989 \\
        & AutoInt        & 0.7483 & 0.6405 & 0.1136 & 0.7022 & 0.5153 & 0.7439 & 0.3952 \\
        & FiBiNet++      & 0.7508 & 0.6437 & 0.1133 & 0.7485 & 0.4868 & 0.7410 & 0.3965 \\
        \midrule
        \multirow{5}{*}{Transformer-based}
        & HiFormer       & 0.7497 & 0.6417 & 0.1134 & \underline{0.7570} & \underline{0.4811} & 0.7248 & 0.4842 \\
        & FAT            & 0.7559 & 0.6490 & \underline{0.1126} & 0.7551 & 0.4828 & 0.7440 & 0.3955 \\
        & RankMixer      & \underline{0.7587} & \underline{0.6526} & 0.1127 & 0.7484 & 0.4859 & 0.7450 & 0.3951 \\
        & OneTrans       & 0.7555 & 0.6497 & 0.1127 & 0.7558 & 0.4823 & \underline{0.7461} & 0.3943 \\        
        \midrule
        \multirow{2}{*}{Spectral-based}
        & WhitenRec      & 0.7481 & 0.6425 & 0.1139 & 0.7472 & 0.4888 & 0.7484 & \underline{0.3932} \\
        & FEDIN          & 0.6945 & 0.5675 & 0.1189 & 0.6172 & 0.9568 & 0.6737 & 0.4256 \\
        \midrule
        & \textbf{SpecFormer} & \textbf{0.7611} & \textbf{0.6537} & \textbf{0.1119} & \textbf{0.7587} & \textbf{0.4791} & \textbf{0.7574} & \textbf{0.3902} \\
        \bottomrule
    \end{tabular}}
\end{table*}

\textbf{1) Overall performance comparisons.} 
SpecFormer consistently and significantly outperforms all baseline models across the evaluated datasets in terms of AUC, GAUC, and Logloss. These impressive results validate the effectiveness and generalizability of our proposed spectral-aware architecture. By fundamentally addressing the spectral collapse bottleneck, SpecFormer proves to be a highly effective and robust feature interaction backbone for real-world recommender systems.

\textbf{2) Compared with Traditional Feature Interaction Models.} 
SpecFormer significantly outperforms traditional recommendation architectures. While these traditional models capture low- and high-order feature interactions through well-designed structures, they inherently lack the context-aware modeling capabilities to model complex dependencies. Furthermore, although models like AutoInt attempt to employ early self-attention, they operate directly on raw feature embeddings, making them highly susceptible to parameter fragmentation~\cite{zhu2025rankmixer}. In contrast, SpecFormer utilizes a modernized tokenization paradigm and operates within the spectral attention, unlocking a much larger capacity for expressive feature interactions.

\textbf{3) Compared with Transformer-based baselines.}
While modern Transformer-based models generally outperform traditional architectures, they still hit a noticeable performance ceiling. Notably, models adopting the standard self-attention mechanism (\eg OneTrans) sometimes even underperform well-designed but simpler token-mixing methods (\eg RankMixer). This empirical evidence strongly corroborates our theoretical analysis in \Cref{sec:theory}: when faced with heterogeneous recommendation data, standard attention mechanisms cause the attention maps to be dominated by a few principal singular values. In contrast, SpecFormer surpasses these strong state-of-the-art baselines by modeling feature interactions within a dynamically softened spectral space. The spectral-aware design effectively revitalizes tail features, forcing the model to capture a more balanced and diverse set of spectral signals, thereby breaking the bottleneck caused by the standard attention mechanism.

 \textbf{4) Compared with Spectral-based baselines.}
SpecFormer exhibits substantial superiority over the spectral-based baselines. Specifically, WhitenRec adjusts the spectrum using static SVD and spectral constraints, but normal spectral transformations are insufficient for complex feature interactions. Furthermore, FEDIN performs poorly across the datasets, confirming that simple frequency-domain transformations struggle to capture non-linear attention-based feature interactions. Conversely, SpecFormer intrinsically integrates spectral-aware feature interaction and mitigates the core spectral collapse in attention mechanism, achieving better performance.

\subsection{Ablation Study (RQ2)}
\label{sec:ablation}

We perform the following ablation study to investigate the effects of each component in SpecFormer:

\begin{itemize}[leftmargin=*]
  \item \textbf{1) Learnable Spectral Softening Module.} 
(a) w/o Spectral Softening ($\mathbf{\Sigma}^* = \mathbf{\Sigma}$): We remove the learnable power-law softening and directly use the original singular values to reconstruct the space.

 \item \textbf{2) Value Matrix Design.} 
(b) Value Matrix w/ Projection ($\mathbf{V}_{\mathrm{val}}=\mathbf{H}\mathbf{W}_v$): We apply a standard linear projection to the Value matrix.
(c) Value Matrix w/ Softening ($\mathbf{V}_{\mathrm{val}}=\mathbf{H}^*$): We use the spectrum-softened embeddings as the Value matrix instead of the original design.

 \item \textbf{3) Spectral and Spatial Residuals.} 
(d) w/o Spectral Position Encoding ($\alpha=0$): We remove the Taylor expansion-based spectral residual position encoding bias.
(e) Spectral Position Encoding $\rightarrow$ Identity Matrix ($\mathbf{P}=\textbf{I}$): We replace the learnable Taylor expansion-based spectral mapping matrix with an identity matrix.
(f) w/o Spatial Residual ($\gamma=0$): We remove the spatial residual term $\mathbf{S}_{\mathrm{bias}}$.

 \item \textbf{4) Spectrum-softened Attention Mechanism.} 
(g) w/o Spectrum-softened Attention $\mathbf{S}$ (Only Residuals): We remove the core dynamic spectrum-softened attention matrix $\mathbf{S}$, leaving only the structural residual biases.
\end{itemize}
As shown in Table~\ref{tab:ablation}, we can observe that:

1) Removing the \textit{Learnable Spectral Softening} leads to a catastrophic performance decline, particularly on the most important AUC metric. This directly validates our core motivation (\cf \Cref{sec:theory}): without dynamically flattening the singular value distribution, the attention mechanism immediately collapses into the dominant principal components, losing its ability to capture long-tail feature interactions.

2) Regarding the \textit{Value Matrix Design}, both (b) and (c) yield severely sub-optimal performance. Applying a standard projection $\mathbf{W}_v$ to the original collapsed input $\mathbf{H}$ is unnecessary and leads to sub-optimal performance. Similarly, using the softened representation $\mathbf{H}^*$ as the Value matrix artificially distorts the semantics of the embedding and might cause information loss. This exploration empirically justifies our design in \Cref{sec:spec_attn}: utilizing the original $\mathbf{H}$ as the Value matrix effectively preserves the raw information, while the softened Query and Key matrices model spectral-aware feature interactions.

3) The ablation on \textit{Residual Position Encodings} demonstrates the necessity of preserving original spectral and spatial signals. Removing the Spectral Position Encoding causes a noticeable drop, proving that while spectral softening revitalizes tail features, it inevitably penalizes principal spectral patterns. The Spectral Position Encoding is crucial as a structural residual to compensate for these essential signals. Furthermore, replacing it with an Identity Matrix results in one of the most severe performance drops. A completely uniform distribution fails to anchor the interactions to the maximum spectrum ($\sigma_1^2$), making it hard for the model to differentiate the varying importance of distinct spectral components. Finally, removing the Spatial Residual slightly degrades performance, highlighting its role in stabilizing training from the original spatial perspective.

4) Removing the \textit{Spectrum-softened Attention} $\mathbf{S}$ leads to a significant decline across all datasets. This indicates that while the residual biases provide excellent structural priors, they cannot replace the dynamic and data-dependent feature interaction capabilities of the  attention mechanism. SpecFormer achieves state-of-the-art performance exclusively through the synergistic combination of the core spectrum-softened attention mechanism and the structured residuals.

\begin{table*}[t]
    \centering
    \caption{Ablation study of SpecFormer. }
    \label{tab:ablation}
    \resizebox{\textwidth}{!}{
    \begin{tabular}{lccccccc}
        \toprule
        \multirow{2}{*}{\textbf{Method}} & \multicolumn{3}{c}{\textbf{Industrial}} & \multicolumn{2}{c}{\textbf{Criteo}} & \multicolumn{2}{c}{\textbf{Avazu}} \\
        \cmidrule(lr){2-4} \cmidrule(lr){5-6} \cmidrule(lr){7-8}
        & AUC $\uparrow$ & GAUC $\uparrow$ & Logloss $\downarrow$ & AUC $\uparrow$ & Logloss $\downarrow$ & AUC $\uparrow$ & Logloss $\downarrow$ \\ 
        \midrule
        \textbf{SpecFormer} & \textbf{0.7611} & 0.6537 & 0.1119 & \textbf{0.7587} & \textbf{0.4791} & \textbf{0.7574} & \textbf{0.3902} \\
        \midrule
        (a) w/o Spectral Softening ($\mathbf{\Sigma}^* = \mathbf{\Sigma}$) & 0.7576 & 0.6541 & 0.1100 & 0.6555 & 0.5332 & 0.7525 & 0.3914 \\
        \midrule    
        (b) Value Matrix w/ Projection ($\mathbf{V}_{\mathrm{val}}=\mathbf{H}\mathbf{W}_v$)                        & 0.7555 & 0.6526 & 0.1103 & 0.6593 & 0.5351 & 0.5645 & 0.4476 \\
        (c) Value Matrix w/ Softening ($\mathbf{V}_{\mathrm{val}}=\mathbf{H}^*$)                                & 0.7547 & 0.6518 & 0.1100 & 0.6472 & 0.5346 & 0.5643 & 0.4479 \\      
        \midrule
        (d) w/o Spectral Position Encoding ($\alpha=0$)                                   & 0.7588 & \textbf{0.6555} & \textbf{0.1099} & 0.6598 & 0.5299 & 0.6700 & 0.4169 \\
        (e) Spectral Position Encoding $\rightarrow$ Identity Matrix ($\textbf{P}=\textbf{I}$)                      & 0.7573 & 0.6548 & 0.1101 & 0.7536 & 0.4832 & 0.5667 & 0.4469 \\
        (f) w/o Spatial Residual ($\gamma=0$)                                     & 0.7551 & 0.6527 & 0.1106 & 0.7555 & 0.4812 & 0.6733 & 0.4203 \\
        \midrule
        (g) w/o Spectrum-softened Attention $\mathbf{S}$ (Only Residuals) & 0.7582 & 0.6551 & 0.1100 & 0.7569 & 0.4812 & 0.7382 & 0.3980 \\
        \bottomrule
    \end{tabular}}
\end{table*}

\subsection{Scaling Study (RQ3)}
\label{sec:scaling}
To evaluate the scalability of \text{SpecFormer} against other Transformer-based models on the depth of Transformer layers, we conducted a comparative study against the representative models \text{RankMixer} and \text{OneTrans} by scaling the model depth from 3 to 5, 7, and 9 layers on the Industrial dataset. We analyzed the performance variations through two key dimensions: the incremental AUC gains ($\Delta$ AUC) and the layer-wise fractional effective rank of the attention matrices. As shown in Figure~\ref{fig:deep_analysis}, we observe that:

1) As the number of Transformer layers increases, the AUC of SpecFormer and RankMixer consistently grows, whereas OneTrans exhibits a severe rise-then-fall performance degradation. This  corroborates our theoretical analysis in \Cref{sec:theory}: OneTrans, relying on standard self-attention, inevitably falls into the vicious cycle of attention collapse. Conversely, SpecFormer fundamentally breaks this bottleneck, demonstrating the capacity to continuously scale up on model depth for substantial recommendation performance gains.

2) While both SpecFormer and RankMixer show the scaling capacity, SpecFormer exhibits a steeper AUC growth curve and consistently outperforms RankMixer across 3, 5, 7, and 9 layers. This observation proves that while carefully hand-crafted token-mixing strategies in RankMixer can avoid the attention collapse by simplifying interactions, they inherently lack the expressive power of complex feature modeling. SpecFormer, equipped with the spectral-aware attention mechanism, not only mitigates the collapse but also unleashes the potential of Transformers in recommendation, modeling fine-grained and high-order feature interactions far more effectively than static token-mixing strategies.

3) As model depth increases, the fractional effective rank of SpecFormer's attention matrix exhibits a steady upward trend. In sharp contrast, the attention effective rank of OneTrans experiences violent fluctuations in early layers and irreversible collapse at deeper layers. This provides the empirical proof for our method: through the spectral-aware attention design, SpecFormer effectively breaks the loop of original attention collapse. It ensures that even at deep depths, the attention map retains a high effective rank, successfully capturing feature interactions across diverse spectral components.

\begin{figure}[t]
    \centering
    \includegraphics[width=1\linewidth]{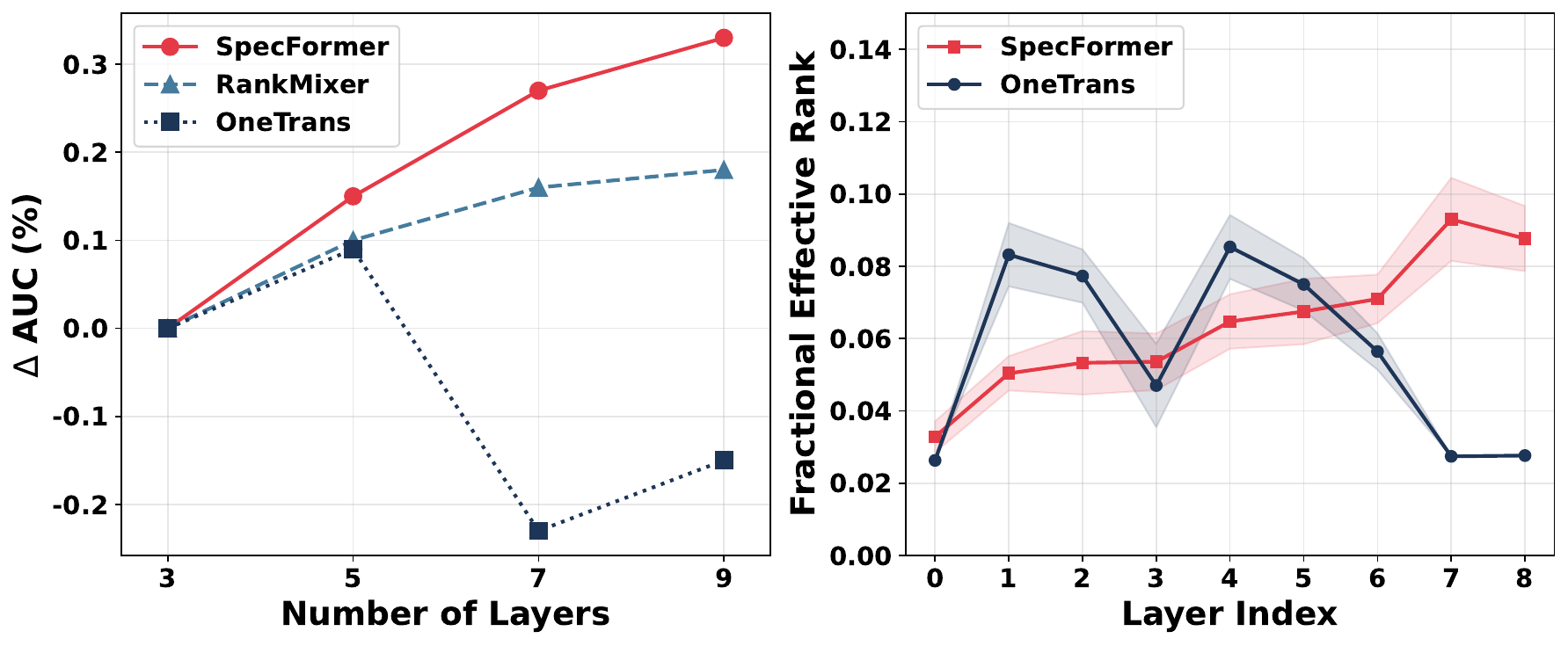}
    \caption{Comparison of AUC gains and attention effective ranks across Transformer layer depth on RankMixer, OneTrans and SpecFormer.
}
    \label{fig:deep_analysis}
\end{figure}

\subsection{Hyperparameter Sensitivity (RQ4)} 
Figure~\ref{fig:hyperpara} showcases the AUC and Logloss performance of the proposed SpecFormer on the Industrial dataset under different spectral and spatial residual position encoding weights (\ie $\alpha$ and $\gamma$). We observe that: 1) As $\alpha$ increases from $0.0001$ to $0.001$ and $0.01$, the recommendation performance initially improves and subsequently declines, generally peaking around $0.001$. Similarly, as $\gamma$ scales across $\{0.5, 0.8, 1.0\}$, the performance exhibits a clear rise-then-fall pattern, achieving optimal results at approximately $0.8$.
This indicates that an appropriate weight of spectral and spatial bias is essential to compensate for the original spectral and attention signals. However, an excessively large $\alpha$ or $\gamma$ will overwhelm the dynamically learned spectrum-softened attention and degrade recommendation performance.
2) Notably, the optimal value of $\alpha$ ($0.001$) is orders of magnitude smaller than that of $\gamma$ ($0.8$). The reasons are: (a) The spatial residual $\mathbf{S}_{\mathrm{bias}}^{(l)}$ is calculated via standard dot-product attention, placing its numerical scale close to the primary spectrum-softened attention matrix $\mathbf{S}^{(l)}$. Consequently, $\gamma$ naturally operates at a scale around $0.8$. (b) The spectral residual term $\mathbf{P}_{\mathrm{bias}}^{(l)}$ is explicitly parameterized by the polynomial expansion of singular values and anchored by  $\sigma_1^2$. Its numerical values are intrinsically much larger than the standard pre-softmax attention scores. Therefore, a strictly small $\alpha$ is required to regularize the magnitude, preventing gradient explosion and ensuring stable training. 

\label{sec:hyperparameters}
\begin{figure}[t]
    \centering
    \includegraphics[width=1.00\linewidth]{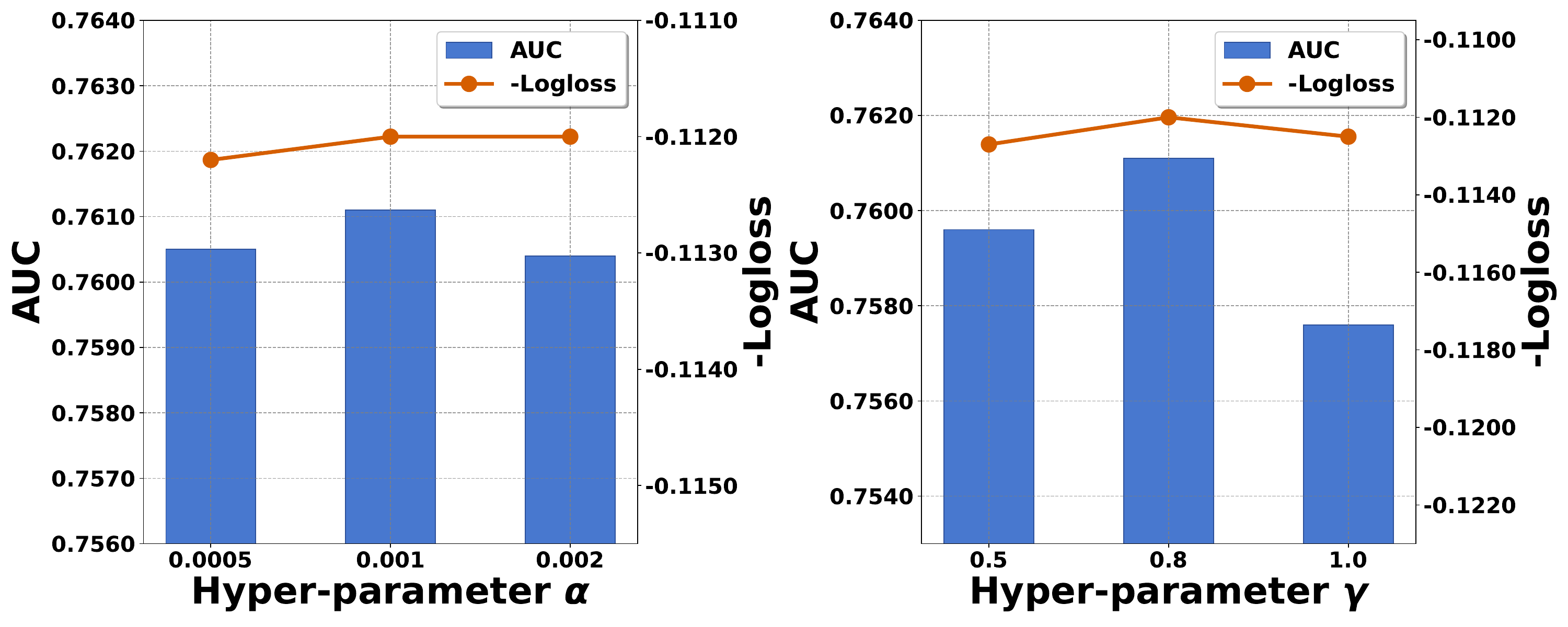}
    \caption{Hyperparameter sensitivity analysis of SpecFormer.
}
    \label{fig:hyperpara}
\end{figure}

\subsection{Online Experiment (RQ5)}
\label{sec:online}
\textbf{Model Scaling and Multi-epoch Training Strategy}. \text{SpecFormer} exhibits superior scalability through its deep architecture, making it highly effective for large-scale parameter expansion. To fully exploit its increased model capacity, we implement a multi-epoch training strategy with a progressive parameter inheritance mechanism. Unlike traditional models like DLRM, which often suffer from rapid performance saturation, \text{SpecFormer} continues to yield significant AUC gains as the number of training epochs increases. This is achieved by utilizing warmed-up embeddings at the start of each epoch and inheriting optimized parameters from the previous one, thereby effectively amplifying the benefits of model scaling as shown in Figure \ref{fig:auc_scaling}. To balance training costs and efficiency, we deployed the 3-epoch model online, achieving a significant 0.92\% AUC gain over the production DLRM baseline.
\begin{figure}[t]
    \centering
    \includegraphics[width=0.8\linewidth]{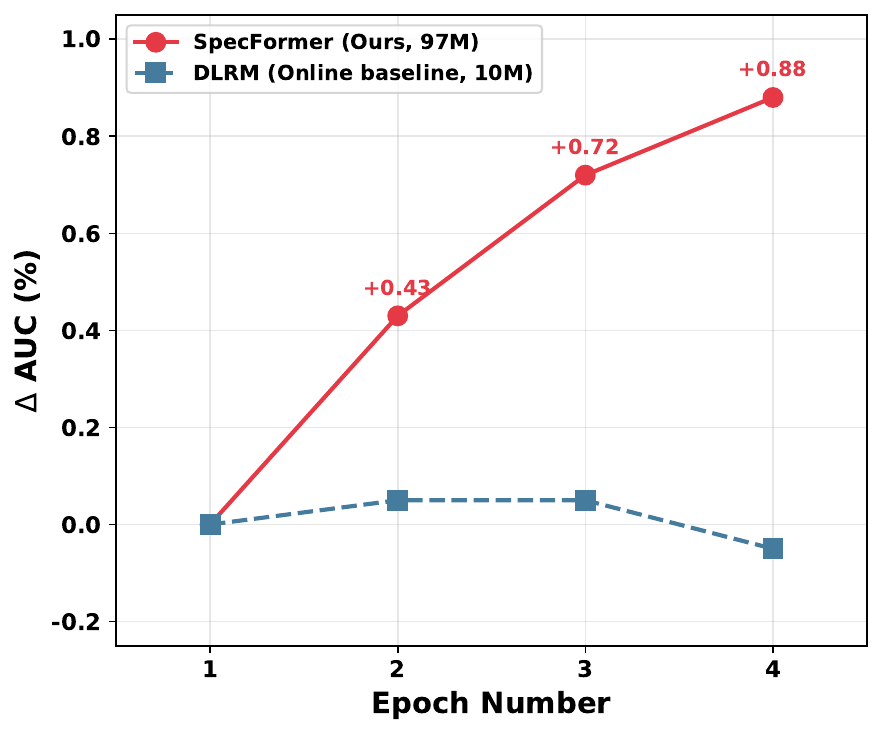}
    \caption{The scaling capacity of SpecFormer against online DLRM baseline under multi-epoch training.}
    \label{fig:auc_scaling}
\end{figure}

\textbf{Efficiency Optimization.} To address the computational complexity of SVD, we employ an in-batch SVD operator that balances spectral refinement with training throughput. On the system level, we integrate a Distributed Dynamic Embedding (DDE) framework, where embedding tables are sharded across different ranks. This allows for dynamic expansion of the feature space while significantly reducing memory footprint. To mitigate communication bottlenecks caused by high-dimensional features, we utilize coalesced lookup operations, successfully reducing embedding query frequency to 8\% of the baseline. For computational efficiency, we integrate FlashAttention, which reduces per-layer attention overhead by 40\%. Collectively, these optimizations, combined with FP16 mixed-precision training, deliver an overall 35\% improvement in training efficiency.

\textbf{Online Performance}. We conducted a large-scale A/B test on Alibaba’s e-commerce advertising platform to evaluate the real-world efficacy of \text{SpecFormer}. The experiment was carried out over a 10-day period from June 1 to June 10, 2026, involving 10\% of the platform’s production traffic. Users were randomly assigned to either the control group (a highly optimized DLRM baseline) or the experimental group (\text{SpecFormer}) using an orthogonal hashing protocol. As summarized in Table \ref{tab:online-performance}, \text{SpecFormer} achieved significant improvements across all key business metrics, including CTR (+1.34\%), and Order (+16.72\%), at the cost of merely a 5ms increase in inference latency, which remains well within the strict online service requirements.
\begin{table}[t]
\centering
\caption{Results of the online A/B experiment on a large-scale e-commerce platform. All performance gains are statistically significant with $p < 0.05$.}
\label{tab:online-performance}
\begin{tabular}{lcccc}
\toprule
\textbf{Method} & \textbf{CTR} & \textbf{CVR} & \textbf{Order} & $\Delta$ \textbf{Latency} \\ \midrule
DLRM Baseline & - & - & - & - \\
\textbf{SpecFormer} & \textbf{+1.34\%} & \textbf{+15.97\%} & \textbf{+16.72\%} & \textbf{+5ms} \\ \bottomrule
\end{tabular}
\end{table}

\section{Related Work}
\label{sec:related_work}

\subsection{Transformer-based Recommender Systems}
Transformer architectures~\cite{vaswani2017attention} have achieved remarkable success across various domains~\cite{devlin2019bert,liu2021swin}. In RS, early adoptions primarily focused on capturing sequential user behaviors~\cite{kang2018self, sun2019bert4rec,yang2023generic}. Subsequently, Transformers were introduced into CTR prediction to automatically model complex and high-order feature interactions~\cite{song2019autoint,pi2020search}. To explicitly map features into interaction spaces, early methods like AutoInt~\cite{song2019autoint} apply multi-head self-attention directly to raw feature embeddings. Building upon this, recent work explores various structural enhancements to model interaction dependencies more effectively. For instance, HiFormer~\cite{gui2023hiformer} utilizes a hierarchical Transformer to handle both fine-grained intra-feature and coarse-grained inter-feature relationships. Furthermore, to optimize interaction efficiency and alleviate parameter fragmentation, modern architectures refine the tokenization and aggregation processes. RankMixer~\cite{zhu2025rankmixer} integrates a specialized token-mixing attention tailored for ranking tasks, and OneTrans~\cite{zhang2026onetrans} formulates a unified attention framework to seamlessly incorporate multiple feature interaction paradigms.
Besides, LLMs also demonstrate strong potential in RS, either functioning as direct  recommenders~\cite{bao2023tallrec, chen2024hllm,lin2024rella,geng2024breaking} or acting as semantic enhancers for recommendation models~\cite{li2025ctrl, xi2024towards, sun2024large,cui2026field}. However, due to their prohibitive inference latency, deploying LLMs in large-scale industrial systems remains highly challenging, and they are not the focus of this paper.

Despite their success, the attention mechanism suffers from suboptimal performance due to the heterogeneity and long-tail distribution characteristics of recommendation data in RS. Existing attempts try to address this with multi-parameters~\cite{guo2023embedding, chen2026rankup} or external spectral enhancement~\cite{zhang2024id,dai2026fedin}, but they either rely heavily on spatial-domain patches or fail to integrate the core attention mechanism, resulting in a limited improvement. In contrast, SpecFormer fundamentally uncovers the limitation of attention mechanism and resolves the performance bottleneck from a spectral perspective.

\subsection{Spectral Collapse in Recommender Systems}
Embedding collapse refers to the phenomenon where representations degenerate into a low-dimensional subspace~\cite{noci2022signal, wei2024diff}, which is a problem extensively studied in Natural Language Generation~\cite{gao2019representation} and Contrastive Learning~\cite{jing2021understanding, hua2021feature}. However, as recent studies indicate~\cite{lin2025recommendation, chen2024towards}, the extreme heterogeneity and long-tail distribution of recommendation data drastically exacerbate this collapse, resulting in a highly skewed embedding space dominated by merely a few principal singular values~\cite{cui2026spectran,hu2025alphafuse}. Meanwhile, attention collapse occurs when the attention map loses its token-routing discriminability, severely limiting the expressive power of self-attention~\cite{bhojanapalli2020low, dong2021attention}. Our theoretical analysis (\cf \Cref{sec:theory}) advances this line of research by mathematically formulating the vicious cycle between embedding collapse and attention collapse during both forward and backward propagation, a phenomenon specifically severe in recommendation scenarios.
Recent works in recommendation introduce explicit spectral regularizations~\cite{zhang2024id} to prevent dimension collapse, and explore Fourier transforms~\cite{dai2026fedin} to capture frequency-domain signals for Transformer-based feature interactions. 
However, these methods typically treat spectral decomposition either as a post-hoc regularization term or as a static transformation isolated from the core attention module. SpecFormer intrinsically mitigates the embedding and attention collapse with the novel spectral-aware attention, which effectively breaks the vicious cycle of collapse in recommender systems.


\section{Conclusion}
\label{sec:conclusion}

In this paper, we investigate the bottleneck of Transformer's attention mechanism in recommender systems. Through empirical and theoretical analysis, we reveal that the intrinsic heterogeneity and long-tail distribution of recommendation data trigger a vicious cycle of spectral collapse in original attention: the collapsed embeddings induce smoothing and low-rank attentions, which in turn act as low-pass filters that starve minor spectral components of gradients during backpropagation. To fundamentally break this bottleneck, we propose SpecFormer, a novel Spectral-Aware Transformer. By introducing a \textit{Learnable Spectral Softening} module, modeling interactions via a \textit{Spectrum-softened Attention}, and the \textit{Spectral Residual Position Encoding}, SpecFormer effectively prevents principal spectral components from dominating the attention map. Extensive offline evaluations and online A/B testing  demonstrate that SpecFormer not only achieves state-of-the-art performance but also unlocks the depth-scaling potential of Transformers for recommendation.



\bibliographystyle{IEEEtran}
\bibliography{ref}

\end{document}